\documentclass[11pt,a4paper]{article}
\usepackage[margin=2cm]{geometry}
\usepackage{setspace}
\usepackage{comment,enumerate}
\usepackage{undertilde}
\usepackage{amsfonts}
\usepackage{amssymb,amsmath,amsthm,bbold}
\usepackage[mathscr]{euscript}

\usepackage{tikz,tikz-cd}
\usepackage[pdfpagelabels]{hyperref}
\hypersetup{
	breaklinks=true,   
	colorlinks=true,   
	pdfusetitle=true,  
}
\newtheorem{thm}{Theorem}[section]
\newtheorem{lemma}[thm]{Lemma}
\newtheorem{cor}[thm]{Corollary}

\theoremstyle{definition}
\newtheorem{defn}[thm]{Definition}
\newcommand{\dvol}{\textnormal{dvol}}

\newcommand{\id}{\mathrm{id}}

\newcommand{\Af}{\mathscr{A}}
\newcommand{\Bf}{\mathscr{B}}
\newcommand{\Cf}{\mathscr{C}}
\newcommand{\If}{\mathscr{I}}
\newcommand{\Uf}{\mathscr{U}}

\newcommand{\Ac}{\mathcal{A}}
\newcommand{\Bc}{\mathcal{B}}
\newcommand{\Cc}{\mathcal{C}}

\newcommand{\dd}{\mathrm{d}}
\newcommand{\ee}{\mathrm{e}}
\newcommand{\ii}{\mathrm{i}}
\newcommand{\dt}{\dd t}
\newcommand{\ds}{\dd s}

\newcommand{\1}{\mathbb{1}}
\newcommand{\CC}{\mathbb{C}}	
\newcommand{\NN}{\mathbb{N}}	
\newcommand{\RR}{\mathbb{R}}

\DeclareMathOperator{\ch}{ch}
\DeclareMathOperator{\Eff}{Eff}
\DeclareMathOperator{\supp}{supp}
\DeclareMathOperator{\Pb}{Prob}

\DeclareMathOperator{\Lin}{Lin}

\newcommand{\Prob}[2]{\Pb(#1;#2)}
	
	\title{Lectures on measurement in quantum field theory \\ {\large EMS-IAMP Spring School -- Symmetry and Measurement
			in Quantum Field Theory}}
	\author{Christopher J Fewster\thanks{\tt chris.fewster@york.ac.uk}\\[0.1cm] 
	{\small Department of Mathematics, University of York, Heslington, York YO10 5DD, United Kingdom.}}
	\date{April 7--11, 2025}
\begin{document}
		\maketitle
	
	\begin{abstract}
		These lectures present a brief introduction to measurement theory for QFT in possibly curved spacetimes introduced by the author and R.~Verch [Comm.\ Math.\ Phys.\ {\bf 378} (2020) 851--889]. Topics include: a brief introduction to algebraic QFT, measurement schemes in QFT, state updates, multiple measurements and the resolution of Sorkin's `impossible measurement' problem. Examples using suitable theories based on Green hyperbolic operators are given, and the interpretational significance of the framework is briefly considered.
		The basic style is to give details relating to QFT while taking for granted various facts from the theory of globally hyperbolic spacetimes. 
	\end{abstract}


\section*{Introduction and scope}

Undergraduate courses in quantum mechanics typically devote considerable attention to the rules of measurement, however this is typically completely absent from courses and texts on quantum field theory (QFT). 
Indeed, until recently there was relatively little attention to the issue at all, and perhaps it is not an exaggeration to say that the available literature was more focussed on potential pathologies and paradoxes than on solutions. The historical reasons for the lack of a good measurement theory for QFT are traced in~\cite{FraserPapageorgiou:2023}, while some examples of the problems encountered are described  in~\cite{EncycMP_Measurement_in_QFT_FewsterVerch2025}.  

The aim of these lectures is to describe a recent proposal developed by the author and Rainer Verch~\cite{FewVer_QFLM:2018}, now sometimes called the Fewster--Verch or FV measurement framework, that 
provides an operational approach to measurement in QFT which is demonstrably causal and covariant, and applies in both flat and curved spacetimes. The aim is to give many of the technical details of~\cite{FewVer_QFLM:2018}, and some of the papers that followed it, in a compact form. As the material was intended for four one-hour lectures, something had to give, and details concerning causal structure of Lorentzian spacetimes are mostly suppressed. The reader who is less versed in such things is free to focus on the example of Minkowski spacetime where the statements are more or less obvious. (But beware: this is \emph{not} to say that any statement that is obviously true in Minkowski spacetime will hold in general Lorentzian spacetimes!)

\paragraph{Acknowledgments} The material in Sections 1--5 was given as lectures at the EMS--IAMP Spring School on Symmetry and Measurement in Quantum Field Theory, at the University of York, April 7--11, 2025. I thank my fellow organisers Daan Janssen and Kasia Rejzner, and the $75+$ participants of the Spring School for their enthusiasm, questions and interest. The author's work is partly supported by EPSRC Grant EP/Y000099/1, which also supported the Spring School. Additional support for the Spring School is gratefully acknowledged from the European Mathematical Society, the International Association of Mathematical Physics, and the \href{
	https://www.cost.eu/actions/CA23115/}{COST Action CA23115}: Relativistic Quantum
Information, funded by COST (European Cooperation in Science and
Technology).

\section{Algebraic QFT by example}	
	
The measurement theory we develop will be expressed in terms of {\bf algebraic quantum field theory} (AQFT). We give a very brief introduction to AQFT, motivated by the examples of the quantum mechanical harmonic oscillator and free Klein--Gordon field.
	
\subsection{Algebraic approach to the harmonic oscillator}
	 
The harmonic oscillator equation
\[ 
	\ddot{q}(t) + \omega^2 q(t) = 0, \qquad q(0)=q,
\]
where $q$ is the standard position operator on $L^2(\RR)$, has the solution (Heisenberg evolution)
\[
q(t) = e^{\ii Ht} q e^{-\ii Ht} =q \cos\omega t + p \frac{\sin\omega t}{\omega}, \qquad\textnormal{where}\quad
H =  \frac{1}{2}p^2 + \frac{1}{2}\omega^2 q^2, \qquad [q,p]=\ii \1,
\]
with $q$ and $p=-\ii \dd/\dd q$ acting as usual on $L^2(\RR)$ -- we suppress operator domains for now.

For any smooth compactly supported complex-valued function $f\in C_0^\infty(\RR)$, define
\[
Q(f) = \int \dt f(t) q(t).
\]
Then it is an \emph{exercise} to check that the following properties hold (e.g., as identities on the Schwartz test functions, which are dense in $L^2(\RR)$):
\begin{enumerate}[H1]
\item the map $f\mapsto Q(f)$ is complex linear [i.e., $Q(\lambda f + \mu g)=\lambda Q(f) + \mu Q(g)$ for $\lambda,\mu\in C_0^\infty(\RR)$, $\lambda,\mu\in\CC$]
\item $Q(f)^*=Q(\overline{f})$
\item $Q(\ddot{f}+\omega^2 f)=0$
\item $[Q(f),Q(g)]=\ii E(f,g)\1$,
\end{enumerate}
where 
\[
E(f,g)= \int \dt\,\ds E(t,s)f(t)g(s), \qquad E(t,s)=-\frac{\sin\omega(t-s)}{\omega}.
\]
It is sometimes convenient to write $Eg$ for the function $(Eg)(t)=\int\ds E(t,s)g(s)$, so that
$E$ becomes a linear operator from $C_0^\infty(\RR)$ to $C^\infty(\RR)$.
Now it is clear that
\[
E(t,s)=E^-(t,s)-E^+(t,s), \qquad\textnormal{where}\quad 
E^\pm(t,s)= \pm \vartheta(\pm(t-s))\frac{\sin\omega(t-s)}{\omega},
\]
where $\vartheta(t)$ is the Heaviside function. The functions $E^\pm$ play an important role in the classical ODE theory of the harmonic oscillator. For $g\in C_0^\infty(\RR)$, one may easily check (\emph{exercise}) that
\[
q(t) = \int \ds\, E^\pm(t,s)g(s) = \int^t_{\mp\infty} \ds\,\frac{\sin\omega(t-s)}{\omega}g(s)
\]
is the unique solution to $\ddot{q}+\omega^2 q = g$ that vanishes to the far past ($+$) or far future ($-$).
That is, $E^+$ is the {\bf retarded Green function} and $E^-$ is the {\bf advanced Green function} for the harmonic oscillator. 

There are good reasons why the difference of the advanced and retarded Green operators appears in the commutator $[Q(f),Q(g)]$, which have to do with {\bf Peierls brackets} -- a grown-up version of 
{\bf Poisson brackets}. 

We can consider the algebra of all polynomial expressions in operators of the form $Q(f)$ and the unit $\1$, using relations 1--4 as rules of simplification. It is an important fact [a version of the Stone--von Neumann theorem] that, modulo modest technical conditions, the \emph{only} irreducible representation of this algebra as Hilbert space operators is the standard quantum mechanical harmonic oscillator (up to unitary equivalence). In this way, almost everything interesting about the system can be bundled up into the algebra of the $Q(f)$'s, subject to the relations H1--4.

\subsection{The free scalar field}

We generalise the algebraic method to a field theory on possibly curved spacetime. Here, an $n$-dimensional {\bf spacetime} is a smooth $n$-manifold $M$ with  Lorentzian metric $g$ of signature ${+}{-}\cdots{-}$ which is oriented and time-oriented. For any set $S\subset M$
\[
J^\pm(S) = \{p\in M: \text{$p$ is reached from $S$ by a piecewise-smooth future/past-directed curve}\},
\]
with the convention $S\subset J^\pm(S)$; we also denote $J(S)=J^+(S)\cup J^-(S)$. The {\bf causal hull} $\ch(S)$ of $S$ is $J^+(S)\cap J^-(S)$.
\begin{figure}
\begin{center}
	\begin{tikzpicture}[scale=0.9]
		\draw[fill=red!20,red!20] (-1,0)--(1,0)--++(2,2)--++(-6,0)--cycle;
		\draw[fill=blue!20] (-1,0).. controls (0,-0.5).. (1,0).. controls (0,0.5).. (-1,0); 
		\node at (0,0){$S$};		
		\node at (0,1){$J^+(S)$};
		\begin{scope}[xshift = 4cm]
		\draw[fill=red!20,red!20] (-1,0)--(1,0)--++(2,-2)--++(-6,0)--cycle;
		\draw[fill=blue!20] (-1,0).. controls (0,-0.5).. (1,0).. controls (0,0.5).. (-1,0); 
		\node at (0,0){$S$};
		\node at (0,-1){$J^-(S)$};
		\end{scope}
		\begin{scope}[xshift = 12cm]
			\draw[fill=red!20,red!20] (-1,0)--(1,0)--++(2,2)--++(-6,0)--cycle;
			\draw[fill=red!20,red!20] (-1,0)--(1,0)--++(2,-2)--++(-6,0)--cycle;
			\draw[fill=blue!20] (-1,0).. controls (0,-0.5).. (1,0).. controls (0,0.5).. (-1,0); 
			\node at (0,0){$S$};
			\node at (0,1){$J(S)$};
		\end{scope}
	\end{tikzpicture}
\end{center}
\begin{center}
	\includegraphics{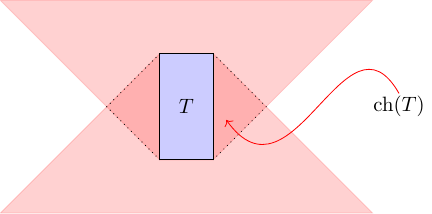}
\end{center}
\caption{Illustration of $J^\pm(S)$, $J(S)$ and an example of a causal hull. Light rays are at $45^\circ$. The set $S$ is causally convex, i.e., equal to its causal hull.}
\end{figure}
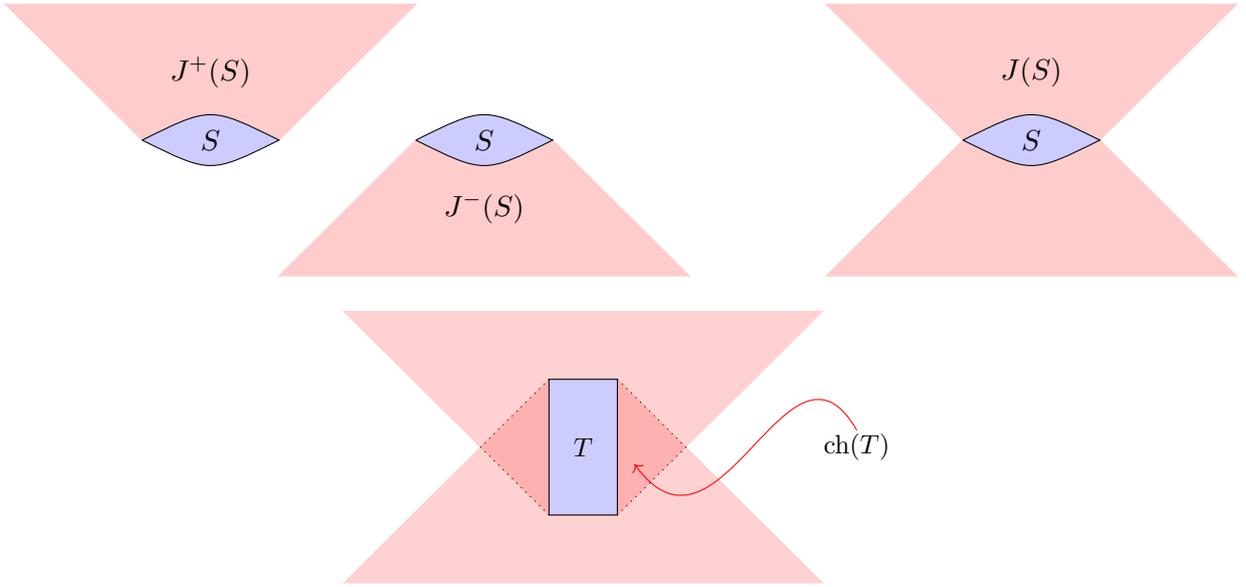
A spacetime is {\bf globally hyperbolic} if it contains no closed causal curves and every compact set has a compact causal hull.
Examples include Minkowski, de Sitter, and Schwarzschild. Anti de Sitter is not globally hyperbolic. 

{\noindent\bf Fact:} On any globally hyperbolic spacetime $M$,
the Klein--Gordon operator $P=\Box+m^2$ ($m\ge 0$) has unique 
Green operators  
$E^\pm: C_0^\infty(M)\to C^\infty(M)$
such that
\begin{enumerate}[G1]
	\item $E^\pm Pf=f$
	\item $PE^\pm f=f$ 
	\item $\supp E^\pm f\subset J^\pm(\supp f)$ 
\end{enumerate}
for all $f\in C_0^\infty(M)$. As before, $E^+$ is the retarded Green operator, while $E^-$ is the advanced Green operator.
 We define the advanced-minus-retarded operator $E=E^--E^+$ and a bilinear pairing 
\[
E(f,h) = \int f E h\,\dvol.
\]

We now quantise the theory by analogy with the harmonic oscillator. Define a unital $*$-algebra $\Af(M)$ with generators $\Phi(f)$ labelled by test functions $f\in C_0^\infty(M)$, and obeying relations
\begin{enumerate}[Q1]
	\item $f\mapsto \Phi(f)$ is complex-linear
	\item $\Phi(f)^*=\Phi(\overline{f})$
	\item $\Phi(Pf)=0$
	\item $[\Phi(f),\Phi(h)]=\ii E(f,h)\1$
\end{enumerate}
for all $f,h\in C_0^\infty(M)$. The interpretation is that $\Phi(f)$ is the integral of the `quantum field' against the test function $f$. It is a {\bf fact} that $\Af(M)$ is nontrivial. If $A\in\Af(M)$ obeys $A\in\Af(M;O)$ we say $A$ is {\bf localisable in $O$.} Clearly it is possible to be localised in very many regions.

\paragraph{Exercises} 
Show that $E$ obeys  $E(Pf,h)=E(f,Ph)=0$ and is antisymmetric (hint - use uniqueness of $E^\pm$). This is necessary for the consistency of relations Q3 and Q4. Show also that $E(f,h)=0$ if $\supp f$ and $\supp h$ are causally disjoint. 
What property of $E$ is a necessary consequence of Q2 and Q4? Prove it directly in the case of the Klein--Gordon equation (hint - start by
showing $E^\pm \overline{f}=\overline{E^\pm f}$).

\paragraph{Local algebras} Any causally convex [i.e., equal to its causal hull] open subset will be called a {\bf region}. Any region is globally hyperbolic in its own right.
To each region $O\subset M$, we assign the subalgebra
\[
\Af(M;O)\subset \Af(M)
\]
generated by $\Phi(f)$ with $f\in C_0^\infty(M)$.
Evidently $\Af(M;M)=\Af(M)$.

\begin{thm}\label{thm:timeslice} 
	For regions $O$, $O_1$ and $O_2$, one has
	\begin{enumerate}
		\item  if $O_1\subset O_2$ then $\Af(M;O_1)\subset\Af(M;O_2)$
		{\bf (isotony)}
		\item if $O_1$ is causally disjoint from $O_2$ (i.e., $O_1\cap J(O_2)=\emptyset$) then $\Af(M;O_1)$ commutes with $\Af(M;O_2)$ {\bf (Einstein causality)}
		\item if $O_1\subset O_2$ and $O_1$ contains a Cauchy surface of $M$ then $\Af(M;O_1)=\Af(M;O_2)$ {\bf (timeslice)}.
	\end{enumerate}
\end{thm}	
\begin{proof}
	1 -- obvious; 2 -- follows from the exercises about $E$. We prove 3 in the case $O_2=M$, writing $O=O_1$. See Fig.~\ref{fig:timeslice} for an illustration. Let $\Sigma^\pm$ be $M$-Cauchy surfaces contained in $O$ with $\Sigma^\pm \subset M\setminus J^\mp(\Sigma^\mp)$.  Choose $\chi\in C^\infty(M)$ 
	so that 
	\[
	\chi(x) =\begin{cases} 1 & x\in J^-(\Sigma^-)\\ 0 & x\in J^-(\Sigma^+).
	\end{cases}
	\] 
	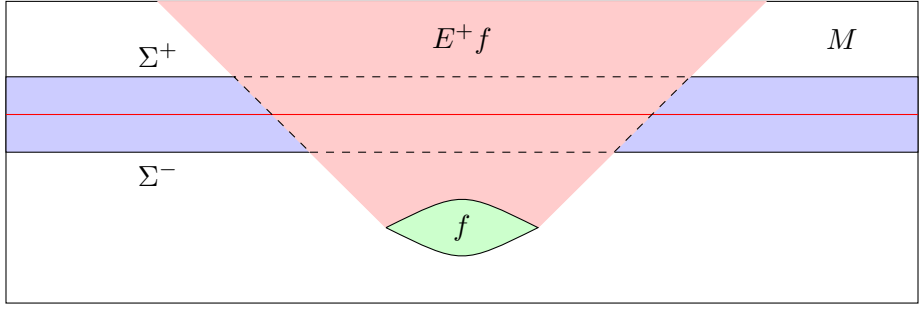
\begin{figure}
		\begin{center}
			\begin{tikzpicture}
				\draw (-6,-1)--++(12,0)--++(0,4)--++(-12,0)--cycle;
				\draw[fill=blue!20] (-6,1)--++(12,0)--++(0,1)--++(-12,0)--cycle;
				\draw[red] (-2,1)--++(4,0);
				\draw[fill=red!20,red!20] (-1,0)--(1,0)--++(3,3)--++(-8,0)--cycle;
				\draw[fill=green!20] (-1,0).. controls (0,-0.5).. (1,0).. controls (0,0.5).. (-1,0); 
				\draw[red] (-6,1.5)--++(12,0);
				\node at (0,0){$f$};		
				\node at (0,2.5){$E^+f$};
				\draw[dashed] (-2,1)--++(4,0)--++(1,1)--++(-6,0)--cycle;
				\node at (5,2.5){$M$};
				\node[above] at (-4,2){$\Sigma^+$};
				\node[below] at (-4,1){$\Sigma^-$};
			\end{tikzpicture}
		\end{center}
		\caption{Illustrating the proof of Theorem~\ref{thm:timeslice}. The red line represents a Cauchy surface within a region $O$ (not shown) that contains the region between Cauchy surfaces $\Sigma^\pm$. The support of $\chi E^+f$ is contained within the dashed boundary and is compact. In this instance $(1-\chi)E^-f=0$.\label{fig:timeslice}}
	\end{figure}
	Then  
	\[
	\chi E f = E^- f -\underbrace{(1-\chi )E^- f}_{\text{compact support}} - \underbrace{\chi E^+ f}_{\text{ditto}}
	\] 
	using facts about globally hyperbolic spacetimes and Green operators\footnote{$\chi$ has future-compact support; $1-\chi$ has past-compact support; $E^\pm f$ and $Ef$ have spatially compact support; and the intersection of a future/past-compact set with a spatially-compact set is compact.} so 
	\[
	[P,\chi]Ef = P\chi E f = f + Ph \in C_0^\infty(M)
	\]
	for some $h\in C_0^\infty(M)$. Meanwhile, $[P,\chi]$ is a partial differential operator with coefficients supported in $\supp\dd \chi\subset O$ and so the LHS is also supported in $O$.
	In the QFT we have
	\[
	\Phi([P,\chi]Ef) = \Phi(f+Ph) =\Phi(f) + \Phi(Ph)=\Phi(f),
	\]
	so every generator of $\Af(M)$ belongs to $\Af(M;O)$ -- hence $\Af(M;O)=\Af(M)$. 
\end{proof}

\clearpage

\section{Measurement schemes}

\begin{defn} 
	A QFT $\Af$ on globally hyperbolic spacetime $M$ consists of a unital $*$-algebra $\Af(M)$ and subalgebras $\Af(M;O)$ labelled by regions $O$ of $M$, subject to isotony, Einstein causality and the timeslice axiom. Recall that a state on $\Af(M)$ is a linear functional $\omega:\Af(M)\to\CC$ obeying $\omega(\1)=1$ and
	$\omega(A^*A)\ge 0$ for all $A\in\Af(M)$. 
\end{defn}
 
Our aim is to develop a framework for measurement on a QFT $\Af$, by modelling the measurement process. We will proceed in an abstract AQFT framework without specifying the nature of the theories under consideration. Let $\Af$ be the {\bf system} or {\bf target} theory. Let $\Bf$ be another QFT that will be a {\bf probe}. Then there is an
{\bf uncoupled combination} $\Uf=\Af\otimes\Bf$ with 
local algebras
\[
\Uf(M;O) = \Af(M;O)\otimes \Bf(M;O) \qquad \textnormal{(algebraic tensor product)}.
\]
{\emph{Exercise} -- Show that $\Uf$ is indeed a QFT!}

The broad idea is to couple the probe to the system in a localised spacetime coupling zone with `before' and `after' regions. Probe and system are prepared independently in the `before' region; in the `after' region we measure the probe and try to learn about the system.
The coupled theory has to be something different from the uncoupled combination in general, but we assume it is still a QFT.  
\begin{defn}
	A QFT $\Cf$ is a {\bf coupled combination} of $\Af$ and $\Bf$ 
	with temporally compact\footnote{Temporally compact $=$ contained in a region bounded by Cauchy surfaces.} {\bf coupling zone} $K\subset M$ (often $K$ is chosen compact) if to each region $O\subset 
	M\setminus \ch(K)$ there is an isomorphism $\iota_O:\Uf(M;O)\to\Cf(M;O)$ compatible with isotony in the sense that the diagram 
	\[
		\begin{tikzcd}
		\Uf(M;O_1) \arrow[d]\arrow[r,"\iota_{O_1}"] & \Cf(M;O_1) \arrow[d] \\
		\Uf(M;O_2) \arrow[r,"\iota_{O_2}"] & \Cf(M;O_2) 
	\end{tikzcd}
	\]
	commutes for all such $O_1\subset O_2$, where the unlabelled arrows are inclusions due to isotony.
\end{defn}
Here, an isomorphism of unital $*$-algebras is a unit-preserving invertible $*$-homomorphism.

In particular, let $M^\pm = M\setminus J^\mp(K)$ -- natural out/in regions,\footnote{The fact that $M^\pm$ are regions holds in globally hyperbolic spacetimes -- see~\cite{FewVer_QFLM:2018}} which contain Cauchy surfaces for $M$ (see Fig.~\ref{fig:Mpm}). Then one has chains of isomorphisms
\[
\begin{tikzcd}
	\Uf(M) \arrow[bend right=20, rrr, "\tau^\pm"]& \arrow[l] \Uf(M;M^\pm) \arrow[r, "\iota_{M^\pm}"]
	& \Cf(M;M^\pm)\arrow[r] & \Cf(M),
\end{tikzcd}
\]
in which the unlabelled isomorphisms are equalities from the timeslice property, and $\tau^\pm$ are formed so as to make the diagram commute.
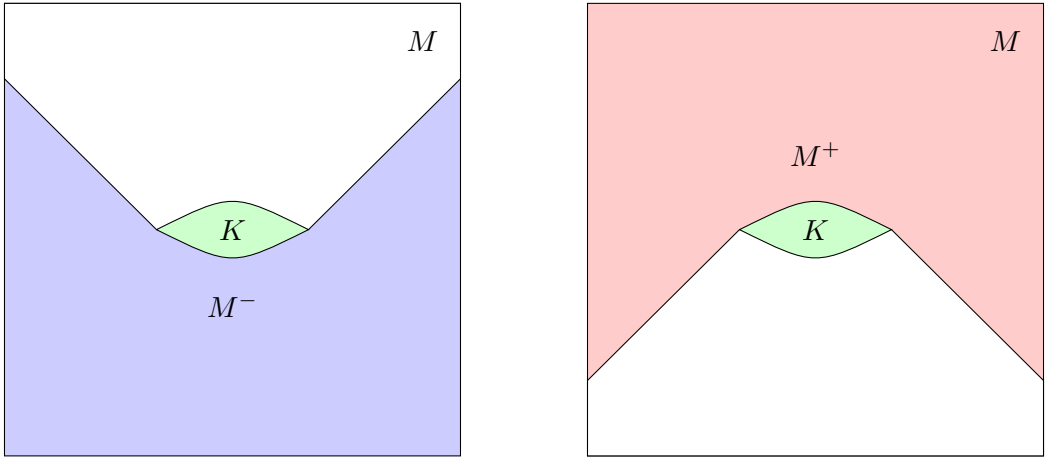
\begin{figure}
\begin{center}
	\begin{tikzpicture}
		\draw[fill=blue!20] (-3,-3)--++(6,0)--++(0,6)--++(-6,0)--cycle;
		\draw[fill=white] (-1,0)--(1,0)--++(2,2)--++(0,1)--++(-6,0)--++(0,-1)--cycle;
		\draw[fill=green!20] (-1,0).. controls (0,-0.5).. (1,0).. controls (0,0.5).. (-1,0);  
		\node at (0,0){$K$};		
		\node at (0,-1){$M^-$};
		\node at (2.5,2.5){$M$}; 
	\end{tikzpicture}\hfil
	\begin{tikzpicture}
		\draw[fill=red!20] (-3,-3)--++(6,0)--++(0,6)--++(-6,0)--cycle;
		\draw[fill=white] (-1,0)--(1,0)--++(2,-2)--++(0,-1)--++(-6,0)--++(0,1)--cycle;
		\draw[fill=green!20] (-1,0).. controls (0,-0.5).. (1,0).. controls (0,0.5).. (-1,0);  
		\node at (0,0){$K$};		
		\node at (0,1){$M^+$};
		\node at (2.5,2.5){$M$}; 
	\end{tikzpicture}
\end{center}
\caption{Diagrams showing the `in' region $M^-$ and `out' region $M^+$ of a compact coupling zone $K$.\label{fig:Mpm}}
\end{figure}

\paragraph{Prepare early, measure late} The physical dynamics of the measuring process is given by a coupled combination $\Cf$. 
However, $\Uf$ provides a convenient intermediate `fictitious language'
between a natural language description and the true dynamics. 

\begin{center}
\begin{tabular}{|c|c|c|}\hline
	Natural language  & $\Uf$ description & $\Cf$ description \\
	& (fictional) & (true dynamics)  \\[2pt] \hline 
	Prepare system and probe in states
	$\omega$ and $\sigma$ before coupling & $\omega\otimes\sigma$ & $\utilde{\omega}_\sigma=(\omega\otimes\sigma)\circ (\tau^-)^{-1}$\\
	Consider probe observable $B\in\Bf(M)$ & $\1\otimes B$  & $\tilde{B}=\tau^+(\1\otimes B)$ \\
	Expectation of probe observable after coupling & & $\utilde{\omega}_\sigma(\tilde{B})$ \\ \hline
\end{tabular}
\end{center}
Note e.g. that 
\[
\utilde{\omega}_\sigma(\iota_{M^-}X)=(\omega\otimes\sigma)(\tau^-)^{-1}\iota_{M^-}X) = (\omega\otimes\sigma)(X)
\]
for $X\in\Uf(M;M^-)$, showing that the state $\utilde{\omega}_\sigma$ is the correct `in' version of $\omega\otimes\sigma$.

To describe the measurement in terms of the system alone, we seek $A\in\Af(M)$ so that
\begin{equation}\label{eq:expectation_matching}
\utilde{\omega}_\sigma(\tilde{B})=\omega(A) \tag{$*$}
\end{equation}
for all $\omega$. The answer requires a useful gadget (closely related
to a conditional expectation).
\begin{defn}
	Let $\Ac$ and $\Bc$ be unital $*$-algebras, and let $\sigma$ be a state on $\Bc$. We define 
	\[
	\eta_\sigma^{(\Ac)}:\Ac\otimes\Bc\to\Ac
	\]
	by linear extension of
	\[
	\eta_\sigma^{(\Ac)}(A\otimes B)= \sigma(B)A, \qquad A\in\Ac,~B\in\Bc.
	\]
	We drop the superscript in $\eta_\sigma^{(\Ac)}$ if it's clear from context.
\end{defn}

\begin{lemma}
	$\eta_\sigma^{(\Ac)}$ is completely positive.\footnote{An element of $\Ac$ is called {\bf positive} if it is a finite sum of terms of form $Z^*Z$; {\bf complete positivity} of $\eta_\sigma^{(\Ac)}$ means that, for all $k\in\NN$ and all positive $P\in \Ac\otimes \Bc\otimes M_k(\CC^k)$, 
	$(\eta^{(\Ac)}\otimes\id)P$ is positive in $\Ac\otimes M_k(\CC)$.} In addition, one has the following identities
	\[
	\eta_\sigma(\1) = \1, \qquad \eta_\sigma(C^*) = \eta_\sigma(C)^*,\qquad
	A\eta_\sigma(C) = \eta_\sigma((A\otimes \1)C), \qquad
	\eta_\sigma(C)A = \eta_\sigma(C(A\otimes \1)),
	\]
	\[
	(\omega\otimes\sigma)(C)=\omega(\eta_\sigma(C)), \qquad \eta_\sigma^{(\Ac)}(C)\otimes B' = \eta_\sigma^{(\Ac\otimes\Bc')}(C\otimes_2 B')
	\]
	for $C\in\Ac\otimes\Bc$, $B'\in\Bc'$. Here $\otimes_2$ is defined by linear extension of $(A\otimes B)\otimes_2 B'=A\otimes B'\otimes B$.
\end{lemma}
\begin{proof}
	Complete positivity is proved in \cite{FewVer_QFLM:2018}. 
	For the identities, check that they hold on elements of the form $A\otimes B$ \emph{(exercise!)} and extend by linearity.
\end{proof}

\begin{thm}
	The equation~\eqref{eq:expectation_matching} is solved by 
	\[
	A = \varepsilon_\sigma(B):= \eta_\sigma(\Theta (\1\otimes B)),
	\]
	where $\Theta = (\tau^{-})^{-1}\circ \tau^+$ is the {\bf scattering map}. This is the unique solution if $\Ac$ is separated by its states (i.e., $\omega(A)=\omega(A')$ for all states $\omega$ $\implies$ $A=A'$).
\end{thm}
\begin{proof}
	We calculate
	\[
	\omega(\varepsilon_\sigma(B)) = \omega(\eta_\sigma(\Theta (\1\otimes B)))= (\omega\otimes\sigma)((\tau^{-})^{-1}\circ\tau^+ (\1\otimes B)) = \utilde{\omega}_\sigma (\tilde{B}).
	\]
	Second part - true by definition.
\end{proof}
Equation~\eqref{eq:expectation_matching} shows that the expected outcome of the actual experiment reproduces the abstract expectation value of the system observable $A$ in state $\omega$. This is a
{\bf measurement scheme} for $A$, which is also called the
{\bf induced system observable}.

\begin{thm}
	$\varepsilon_\sigma:\Bf(M)\to\Af(M)$ is a completely positive map, obeying
	\[
	\varepsilon_\sigma(\1) =\1, \qquad
	\varepsilon_\sigma(B^*) = \varepsilon_\sigma(B)^*.
	\]
	so self-adjointness is preserved.
	(However, $\varepsilon_\sigma$ is \emph{not} a homomorphism in general).
\end{thm}
\begin{proof}
	Exercises, using the properties of $\eta_\sigma$ and $\Theta$.
\end{proof} 

\paragraph{Localisation}
Historically, the measurement scheme described above was developed to answer the question of what system observable is measured when one couples a probe to a QFT and measures a given probe observable, and to know where it is localisable. The first part is answered by the map $\varepsilon_\sigma$, while the second part is answered as follows. First observe that
\[
	[A,\varepsilon_\sigma(B)] = 
	[A,\eta_\sigma (\Theta (\1\otimes B))]
	= \eta_\sigma([A\otimes \1,\Theta (\1\otimes B)]) .
\]
We introduce the notation $K^\perp = M\setminus J(K)$ for the causal complement of $K\subset M$. We also say that $\Af$ has the
{\bf Haag property} if to any compact $K\subset M$, one has
\[
\{A\in\Af(M): [A,B]=0~\forall B\in \Af(M;K^\perp) \} =:\Af(M;K^\perp)^c\subset
\bigcap_O \Af(M;O),
\]
where the intersection is taken over connected regions containing $K$.
The superscript $c$ denotes a {\bf relative commutant}.
\begin{thm}
	(a) If $A\in\Af(M;K^\perp)$ then  
	\[
	[A,\varepsilon_\sigma(B)] = 0\qquad \forall B\in\Bf(M).
	\]
	Thus if $\Af$ has the Haag property and $K$ is compact, then 
	$\varepsilon_\sigma(B)\subset \Af(M;O)$ for any connected region
	$O$ containing $K$. 
	
	(b) If $B\in\Bf(M;K^\perp)$ then $\varepsilon_\sigma(B)\in \CC\1$.
\end{thm}
\begin{proof}
	Both parts rely on a key property of scattering maps,
	that $\Theta$ acts trivially on $\Uf(M;K^\perp)$ (proved below).
	Given that, (a) uses the calculation
	\[
	[A\otimes\1,\Theta(\1\otimes B)] = \Theta[A\otimes\1,\1\otimes B] = 0,	
	\]
	and (b) uses 
	\[
	\varepsilon_\sigma(B) = \eta_\sigma (\Theta\1\otimes B)=\eta_\sigma(\1\otimes B) = \sigma(B)\1.
	\]
	The second part of~(a) is immediate because we have shown that $\varepsilon_\sigma(B)\in\Af(M;K^\perp)^c$.
\end{proof}
Part~(a) answers the localisation question, while~(b) reassuringly shows that no useful information about the system can be extracted from probe observables that are causally separated from the coupling zone, and should not even know of its existence. 

Now we prove the key property mentioned above.
\begin{lemma}\label{lem:Theta_on_Kperp}
	$\Theta$ acts trivially on $\Uf(M;K^\perp)$.
\end{lemma}
\begin{proof} Consider the diagram
	\[
		\begin{tikzcd}
			\Uf(M)\arrow[bend right=70,ddddrr,swap,"\tau^+"] & & \arrow[ll]\Uf(M;K^\perp)\arrow[dr]\arrow[dl]\arrow[dd,"\iota_{K^\perp}"] \arrow[rr] & &\Uf(M)\arrow[bend left = 70,ddddll,"\tau^-"] \\
			&\Uf(M;M^+) \arrow[ul,"\cong"]\arrow[dd,"\iota_{M^+}"] & & \Uf(M;M^-)\arrow[dd,"\iota_{M^-}"]\arrow[ur,"\cong"]& \\
			& &  \Cf(M;K^\perp)\arrow[dd]\arrow[dr]\arrow[dl] & &\\
			& \Cf(M;M^+)\arrow[dr,"\cong"] & & \Cf(M;M^-)\arrow[dl,"\cong"] &\\
			& & \Cf(M) & &
		\end{tikzcd}	
	\]
in which the arrows labelled with $\cong$ are timeslice isomorphisms, and all unlabelled arrows are isotony inclusions. The two outer lobes
commute by definition of $\tau^\pm$. The triangles all commute by transitivity of inclusion. The two quadrilaterals commute
by compatibility of the isomorphisms $\iota_O$ with isotony. 
So the diagram commutes in full, which means that the restrictions of
$\tau^\pm$ to $\Uf(M;K^\perp)$ agree, which means that $\Theta=(\tau^-)^{-1}\circ\tau^+$ acts trivially on $\Uf(M;K^\perp)$.
\end{proof}

\clearpage

\section{State updates}

\begin{defn}
	An {\bf effect} in a unital $*$-algebra $\Ac$ is $E\in\Ac$ such that
	$E\ge 0$ and $\1-E\ge 0$. (Here $X\ge 0$ iff $X=\sum_\alpha A_\alpha^* A_\alpha$ for some $A_\alpha$'s in $\Ac$, so all effects are self-adjoint.)
	The set of effects is denoted $\Eff(\Ac)$. An effect $E$ is {\bf sharp} iff $E$ is a projection, and {\bf unsharp} otherwise.
\end{defn}
An effect represents a binary test so that
\[
\omega(E) = \Prob{\text{Test of $E$ returns $1$}}{\omega}.
\]
We often say success/fail or pass/fail in place of $1$/$0$.

If $E\in\Eff(\Ac)$ and $F\in\Eff(\Bc)$ then $E\otimes F\in\Eff(\Ac\otimes\Bc)$ (\emph{exercise}) is interpreted as the 
test $E\& F$, because when the  tests are independent one has
\[
(\omega\otimes\sigma)(E\otimes F) = \omega(E)\sigma(F)=
\Prob{E=1}{\omega} \Prob{F=1}{\sigma} = \Prob{(E=1) \& (F=1)}{\omega,\sigma}.
\]

Let $\Af$ and $\Bf$ be system and probe QFTs, with coupled combination $\Cf$ defined as above. Let $E\in\Eff(\Af(M))$, $F\in\Eff(\Bf(M))$ 
and consider a post-coupling test of $E\& F$ when system and probe are prepared in states $\omega$ and $\sigma$ in the pre-coupling zone.
A similar argument to that used when considering the measurement schemes shows that
\[
\Prob{E\& F}{\omega\otimes\sigma}= (\omega\otimes\sigma)(\Theta (E\otimes F))=:\If_\sigma(F)(\omega)(A),
\]
where, for any state $\sigma$ on $\Bf(M)$, the pre-instrument $\If_\sigma:\Eff(\Bc)\to \Lin (\Af(M)^*)$ is defined by 
\[
\If_\sigma(F)(\omega)(A) = (\omega\otimes\sigma)(\Theta (A\otimes F))
=\omega(\eta_\sigma(\Theta (A\otimes F))).
\]
In particular, 
\[
\If_\sigma(F)(\omega)(\1)= \omega(\varepsilon_\sigma(F)) = \Pb(\varepsilon_\sigma(F);\omega).
\]
The classical conditional probability for passing $E$, conditioned on passing $F$ is
\[
\Pb_\sigma(E|F;\omega) = \frac{\Pb_\sigma(E\& F;\omega)}{
\Pb_\sigma(\underbrace{F}_{=\1 \& F};\omega)} = \frac{\If_\sigma(F)(\omega)(E)}{\If_\sigma(F)(\omega)(\1)},
\]
where we acknowledge the dependence on the probe preparation state in the subscript.
Set 
\[
\omega' = \frac{\If_\sigma(F)(\omega)}{\If_\sigma(F)(\omega)(\1)}.
\]
Then $\omega'$ is a state ($\omega'(\1)=1$ by construction; 
as $F\ge 0$ one has $\Theta(A\otimes F)\ge 0$ and hence $\eta_\sigma(\Theta (A\otimes F))\ge 0$ by positivity of $\eta_\sigma$). Clearly,
\[
\omega'(E) = \Pb_\sigma(E|F;\omega) ,
\]
so $\omega'$ is the state that predicts the probability of $E$ conditioned on a success in $F$. Updating the state from $\omega$ to $\omega'$ allows us to compute the conditional probability directly. This, we claim, is the role of a {\bf state update rule}. This particular update is {\bf selective} because it is conditioned on $F$. A {\bf nonselective update} arises by conditioning on the trivially passed test, i.e., $F=\1$, resulting in
\[
\omega'_{\text{ns}} = \If_\sigma(\1)(\omega)
, \qquad \text{i.e.,}\quad
\omega'_{\text{ns}}(A) = (\omega\otimes\sigma)(A\otimes \1).
\]
\emph{Exercise:} Show that 
\[
\omega'_{\text{ns}} = \Pb_\sigma(E;\omega)\omega'_E + \Pb_\sigma(\neg E;\omega)\omega'_{\neg E},
\]
where $\neg E=\1-E$, so it is the update appropriate when a measurement has been made but the outcome is unknown, thus requiring a sum over outcomes weighted by their probabilities. 

Now suppose that $A\in\Af(M;K^\perp)$. Then 
\[
\If_\sigma(F)(\omega)(A)= \omega(\eta_\sigma \Theta (A\otimes F)),
\]
but $\Theta (A\otimes F) = (A\otimes 1)\Theta (1\otimes F)$, so
by the properties of $\eta_\sigma$, 
\[
\If_\sigma(F)(\omega)(A)= \omega (A\eta_\sigma (\Theta (1\otimes F))) = \omega(A\varepsilon_\sigma(F)).
\]
Then the expectation of $A$ in the updated state is
\[
\omega'(A) = \frac{\omega(A\varepsilon_\sigma(F))}{\omega(\varepsilon_\sigma(F))}.
\]
Note that the expectation value in the selectively updated state will differ in general from the original state, despite the fact that
$A$ is localisable spacelike to the coupling region. Indeed, 
$\omega'(A)=\omega(A)$ precisely if $A$ and $\varepsilon_\sigma(F)$ are uncorrelated in the state $\omega$, i.e., 
\[
\omega(A\varepsilon_\sigma(F))=\omega(A)\omega(\varepsilon_\sigma(F)).
\]
There is no contradiction or curiosity here: the same effect can be obtained in purely classical measurements (recall or google the famous example of Bertlmann's socks). One might call it
{\bf unspooky action at a distance} [though there is neither spookiness, nor action...].

But now consider the nonselective situation, again with $A\in\Af(M;K^\perp)$. This is the same as setting $F=\1$, so  
\[
\omega'_{\text{ns}} (A) = \omega(A),
\]
i.e., we have established the {\bf principle of blissful ignorance} that the expectation value is unchanged when the measurement is nonselective. This is just as well: science would be very difficult if one had to know about all experiments in one's causal complement, in order to appropriately describe one's own.

Lastly, a recent result~\cite{Fewster:2025a} is that, in suitable models (cf.\ Section~\ref{sect:examples}) the nonselective update of a Hadamard state is Hadamard. 

\clearpage
\section{Multiple measurements}

If $K_1,K_2$ are temporally compact sets, write $K_1\triangleleft K_2$
if $K_1 \cap J^+(K_2)=\emptyset$ (equivalently $K_2 \cap J^-(K_1)=\emptyset$). In a globally hyperbolic spacetime $K_1\triangleleft K_2$ if and only if there is a Cauchy surface $\Sigma$ so that $K_1$ ($K_2$) lies strictly to the past (future) of $\Sigma$. Note that $K_1$ and $K_2$ are causally disjoint if and only if both $K_1\triangleleft K_2$ and $K_2\triangleleft K_1$. 
\begin{defn}
	A finite collection of temporally compact subsets is {\bf causally orderable} if it can be labelled $K_1,\ldots,K_n$ so that 
	$K_1\triangleleft K_2 \triangleleft\cdots \triangleleft K_n$, 
	in which case we say that $K_1,\ldots,K_n$ are {\bf causally ordered}. A causally orderable collection may admit more than one
	causal order.
\end{defn}
 
Consider causally ordered $K_1\triangleleft K_2$. Suppose $\Af$ is a system QFT and $\Bf_1$ and $\Bf_2$ are probe theories, and that there are coupled combinations $\Cf_j$ of $\Af$ with $\Bf_i$ for $j=1,2$ with coupling zones $K_j$. On the other hand, we could consider 
$\Bf_1\otimes \Bf_2$ as a combined probe theory, and that there is a
coupled combination of $\Af$ and $\Bf_1\otimes \Bf_2$ with coupling
zone $K_1\cup K_2$, which is in some sense a merger of $\Cf_1$ and $\Cf_2$. If all the theories are Lagrangian and the couplings are all
defined by interaction terms then it is clear what we mean by the merged theory -- we simply add the two interaction terms together. 
In general, we need make an assumption of {\bf causal factorisation}
\[
\Theta = (\Theta_1\otimes \id )\circ (\Theta_2\otimes_2\id),
\]
where $\Theta$ is the scattering map of $\Cf$ relative to $\Uf = \Af\otimes (\Bf_1\otimes\Bf_2)$ while $\Theta_j$ is the scattering map of $\Cf_j$ relative to $\Uf_j=\Af\otimes\Bf_j$. Causal factorisation holds in models and is expected in reasonable general models. A stronger additivity property is built in to the Buchholz--Fredenhagen approach to constructing interacting models~\cite{BuchholzFredenhagen:2020}.

\begin{thm}[Composition of pre-instruments] With $K_1\triangleleft K_2$ and assuming causal factorisation, one has
	\[
	\If_{\sigma_2}(E_2)\circ \If_{\sigma_1}(E_1) = \If_{\sigma_1\otimes\sigma_2}(E_1\otimes E_2)
	\]
	for $E_j\in\Eff(\Bf(M))$ ($j=1,2$).
\end{thm}
\begin{proof}
	A calculation, using the definitions and properties of $\eta_\sigma$-maps, together with causal factorisation
	\begin{align*}
		\If_{\sigma_2}(E_2)(\If_{\sigma_1}(E_1)(\omega))(A) &=
		(\If_{\sigma_1}(E_1)(\omega)) (\eta^{(\Af(M))}_{\sigma_2} \Theta_2 (A\otimes E_2) ) \\
		&= (\Theta_1^*(\omega\otimes\sigma_1))((\eta_{\sigma_2}^{(\Af(M))} \Theta_2 (A\otimes E_2))\otimes E_1) \\
		&= (\Theta_1^*(\omega\otimes\sigma_1))(\eta_{\sigma_2}^{(\Af(M)\otimes\Bf_1(M))}( \Theta_2 (A\otimes E_2)\otimes_2 E_1)) \\
		&=((\Theta_1^*(\omega\otimes\sigma_1))\otimes\sigma_2)
		( (\Theta_2\otimes_2\id) (A\otimes E_1\otimes E_2))\\
		&= (\Theta_1\otimes \id)^* (\omega\otimes\sigma_1\otimes\sigma_2)( (\Theta_2\otimes_2\id) (A\otimes E_1\otimes E_2)) \\
		&= (\omega\otimes\sigma_1\otimes\sigma_2)(\Theta (A\otimes E_1\otimes E_2))\\
		&= \If_{\sigma_1\otimes\sigma_2}(E_1\otimes E_2)(\omega)(A).
	\end{align*}
	As $A$ and $\omega$ are arbitrary, the result is proved.
\end{proof}
 
\begin{cor}
	If $K_1$ and $K_2$ are causally disjoint then, for $E_j\in\Eff(\Bf_j(M))$ ($j=1,2$), 
	\[
	\If_{\sigma_2}(E_2)\circ \If_{\sigma_1}(E_1) = \If_{\sigma_1\otimes\sigma_2}(E_1\otimes E_2) = 
	\If_{\sigma_1}(E_1)\circ \If_{\sigma_2}(E_2) . 
	\]
	Consequently, the corresponding selective or nonselective updates can be performed in either order, or in a single combined update.
\end{cor} 
This is a central consistency result, and a major bonus of deriving the update rules from QFT rather than imposing rules motivated from quantum mechanics. 
\begin{cor}
	Suppose $K_1\triangleleft K_2$ and let $E_j\in\Eff(\Bf_j(M))$, $j=1,2$. Then 
	\[
	\Pb_{\sigma_1\otimes\sigma_2}(E_1\&E_2;\omega) = \Pb_{\sigma_2}(E_2;\omega')\Pb_{\sigma_1}(E_1;\omega),
	\]
	or equivalently,
	\[
	\omega(\varepsilon_{\sigma_1\otimes\sigma_2}(E_1\otimes E_2)) = 
	\omega'(\varepsilon_{\sigma_2}(E_2))\omega(\varepsilon_{\sigma_1}(E_1)),
	\]
	where $\omega'$ is the updated system state consequent on a successful test of $E_1$.
\end{cor}
\begin{proof}
	Another calculation:
	\begin{align*}
	\omega(\varepsilon_{\sigma_1\otimes\sigma_2}(E_1\otimes E_2)) &= 
	\If_{\sigma_1\otimes\sigma_2}(E_1\otimes E_2)(\omega)(\1)
	=\If_{\sigma_2}(E_2)(\If_{\sigma_1}(E_1)(\omega))(\1)\\
	&= (\If_{\sigma_1}(E_1)(\omega))(\varepsilon_{\sigma_2}(E_2))\\
	&= \omega'(\varepsilon_{\sigma_2}(E_2)) \omega(\varepsilon_{\sigma_1}(E_1)),
	\end{align*}
	recalling that $\If_{\sigma_1}(E_1)(\omega)=\omega(\varepsilon_{\sigma_1}(E_1))\omega'$.
\end{proof} 
Iterating this result, we see that state updates can be made sequentially from the past to the future.
\begin{cor}
	For a causally ordered set of coupling zones $K_1\triangleleft K_2\triangleleft\cdots\triangleleft K_n$, and assuming causal factorisation,
	\[
	\Pb_{\sigma_1\otimes\cdots\otimes\sigma_n}(E_1\&\cdots \& E_n) = 
	\Pb_{\sigma_n}(E_n;\omega^{(n-1)}) \cdots \Pb_{\sigma_2}(E_2;\omega^{(1)})\Pb_{\sigma_1}(E_1;\omega)
	\]
	for $E_j\in\Eff(\Bf_j(M))$, where $\omega^{(r)}$ is the 
	updated state for a successful selective test of $E_{r}$ in state $\omega^{(r-1)}$ for $r\ge 1$ with $\omega^{(0)}=\omega$.
	The LHS can be computed in any compatible causal order with the same result.
\end{cor}

Recalling that failure of $E_j$ is represented by the effect $\neg E_j=\1-E_j$, one can apply this result to any history of outcomes for the tests. Finally, the formalism extends to multi-outcome {\bf effect valued measures} -- see~\cite{FewVer_QFLM:2018} and~\cite{MandryschNavascues:2024}.

\clearpage
\section{Impossible measurements}

We come to a now well-known example in measurement theory: Sorkin's `impossible measurement' problem~\cite{sorkin1993impossible} -- see the recent survey~\cite{PapageorgiouFraser:2023b}. Consider three regions, $O_A$, $O_B$ and $O_C$, corresponding to the labs of Alice, Bob and Charlie and with causal relations shown in the figure. Specifically: $O_C$ does not intersect $O_B$'s past, and neither does $O_B$ intersect $O_A$'s past; there are nontrivial intersections $J^+(O_A)\cap O_B$ and
$J^+(O_B)\cap O_C$, but $O_C$ and $O_A$ are causally disjoint.

Sorkin used essentially this arrangement to illustrate the pitfalls in trying to straightforwardly adapt rules of quantum measurement from quantum mechanics. The idea is Alice will choose whether to make a nonselective measurement, while Bob will certainly make a nonselective measurement. Can Charlie tell whether Alice has measured, which would amount to a form of superluminal communication? In the rules Sorkin suggested, an apparently typical measurement made by Bob allows Charlie to determine whether Alice has measured. The response is that Bob's measurement cannot be possible -- it is an {\bf impossible measurement}. Sorkin then writes:
\begin{quote}
``[I]t becomes a priori unclear, for quantum field theory, which observables can be
measured consistently with causality and which can't. This would seem to deprive [QFT] of any definite measurement theory,
	leaving the issue of what can actually be measured to (at best) a case-by-case analysis''
\end{quote}
In fact he then argues for a consistent histories interpretation. What we will do here, following~\cite{BostelmannFewsterRuep:2020}, is to show that the problem of impossible measurements is resolved in the measurement framework described so far. Let us suppose that Alice and Bob make their measurements by coupling to probes with coupling zones $K_A\subset O_A$ and $K_B\subset O_B$ respectively.

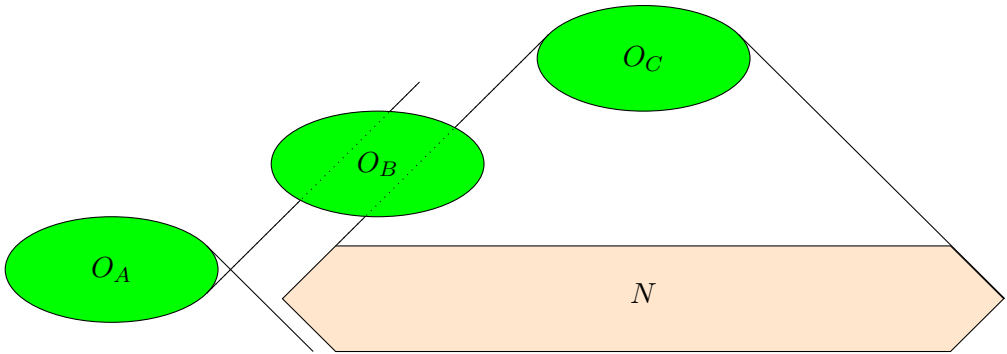
\begin{figure}[b]
	\begin{center}
		\begin{tikzpicture}[scale=0.7] 
			\draw[fill=green] (0,0) node{$O_A$} ellipse (2cm and 1cm);
			\draw(2-0.211,-0.447)--++(4,4);
			\draw(8.211,4.447)--++(-4,-4);
			\draw[fill=green] (5,2) node{$O_B$} ellipse (2cm and 1cm);
			\draw[fill=green] (10,4) node{$O_C$} ellipse (2cm and 1cm);
			\draw[dotted](2-0.211,-0.447)--++(4,4);
			\draw[dotted](8.211,4.447)--++(-4,-4); 
			\draw[fill=orange!20] (3.211,0.447-1) --++(1,-1) --++(12-0.442,0) --++(1,1)--++(-1,1)--++(-12+0.442,0)--cycle;
			\node at (10,-0.447){$N$};
			\draw(12-0.211,4.447)--++(5,-5);
			\draw(2-0.211,0.447)--++(2,-2);	
		\end{tikzpicture}
	\end{center}
	\caption{$O_A$ and $O_C$ are causally disjoint, though $O_A$ can influence $O_B$, and $O_B$ can influence $O_C$. In a globally hyperbolic spacetime it follows that there is a region $N$ in $B$'s `in' region, containing $O_C$ in its Cauchy development but causally disjoint from $O_A$.}
\end{figure}

What is necessary is to compare $\omega_{AB}(C)$ and $\omega_B(C)$, where $C\in\Af(O_C)$, $\omega_B$ is the nonselective update of Bob's measurement alone, and $\omega_{AB}$ the update when both Alice and Bob make nonselective updates. We have
\[
\omega_B(C) = (\omega\otimes\sigma_B)(\Theta_B (C\otimes\1)) = 
\If_{\sigma_B}(\1)(\omega)(C),
\]
while
\begin{align*}
\omega_{AB}(C) &= \If_{\sigma_B}(\1)(\omega_A)(C) = 
\If_{\sigma_B}(\1)(\If_{\sigma_A}(\1)(\omega))(C) =
\If_{\sigma_A\otimes\sigma_B}(\1\otimes\1)(\omega)(C) \\
& = (\omega\otimes\sigma_A\otimes\sigma_B)((\Theta_A\otimes \id)\circ (\Theta_B\otimes_2\id) (C\otimes\1\otimes\1)).
\end{align*}
A geometric argument~\cite{BostelmannFewsterRuep:2020} shows that there is a region $N$ as shown, so that $N$ lies in Bob's `in' region, i.e., $N\subset M\setminus J^+(K_B)$, and is causally disjoint from Alice's coupling zone $K_A$.
Meanwhile,  $O_C$ lies in Bob's `out' region by assumption and is contained in the Cauchy development of $N$.

\begin{figure}
	\begin{center}
		\begin{tikzpicture}
				\draw[fill=green!20] (-1,0).. controls (0,-0.5).. (1,0).. controls (0,0.5).. (-1,0);  
				\node at (0,0) {$K$};
				\draw[dashed] (-1,0) --++(-2,2);
				\draw[dashed] (-1,0) --++(-2,-2);
				\draw[dashed] (1,0) --++(2,2);
				\draw[dashed] (1,0) --++(2,-2);
				\draw[fill=blue!20] (-1,-1.5)--++(0.5,-0.5)--++(5,0)--++(0.5,0.5)--++(-0.5,0.5)--++(-5,0)--cycle;
				\node at (2,-1.5) {$L^-$};
				\draw[dotted] (-0.5,-1)--++(2.5,2.5)--++(2.5,-2.5);
				\draw[fill=orange!20] (2,1.5)--++(0.5,-0.5)--++(-0.5,-0.5)--++(-0.5,0.5)--cycle;
				\node at (2,1) {$L^+$};
		\end{tikzpicture}
	\end{center}
	\caption{Illustration of the sets appearing in Lemma~\ref{lem:Theta_on_Lplus}.
	$L^-$ lies in the `in' region $M^-=M\setminus J^+(K)$ of coupling zone $K$, while $L^+$ lies in the `out' region $M^-=M\setminus J^+(K)$ and is also in the Cauchy development of $L^-$.\label{fig:Theta_on_Lplus}}
\end{figure}
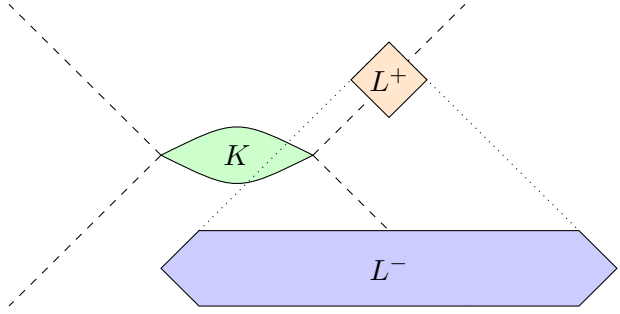

The following is a general result about scattering maps, illustrated in Fig.~\ref{fig:Theta_on_Lplus}.
\begin{lemma}\label{lem:Theta_on_Lplus}
	Suppose, for coupling zone $K$, that $L^\pm\subset M^\pm=M\setminus J^\mp(K)$ are regions with $L^+\subset D(L^-)$, the Cauchy development (aka domain of determinacy) of $L^-$. Then 
	\begin{equation}\tag{$**$}\label{eq:Theta_Lplus}
	\Theta\Uf(M;L^+)\subset \Uf(M;L^-).
	\end{equation}
\end{lemma}
We apply Lemma~\ref{lem:Theta_on_Lplus} with $L^+=O_C$, $L^-=N$, and for Bob's coupling. Then
\[
\Theta_B (C\otimes\1) \in \Uf_B(M;N) \qquad \implies
(\Theta_B\otimes_2 \1 )(C\otimes \1\otimes\1)\in \Uf_{AB}(M;N).
\]
But $\Theta_A\otimes\id$ acts trivially on $\Uf_{AB}(M;N)$ so 
\[
\omega_{AB}(C) = (\omega\otimes\sigma_A\otimes\sigma_B)( (\Theta_B\otimes_2\id) (C\otimes\1\otimes\1)) = \omega_B(C).
\] 
Charlie's results are (in expectation value) the same whether or not Alice measures -- there are no impossible measurements in this framework. 
Put another way, in order to see an impossible measurement, one requires measurement couplings that cannot be modelled using (local) QFTs. \emph{Impossible measurements require impossible apparatus}~\cite{BostelmannFewsterRuep:2020}.

Summarising
\begin{itemize}
	\item there are local algebra elements that can produce superluminal communication if they could be implemented as operations within the coupling zone (Sorkin's argument~\cite{sorkin1993impossible}, and more detailed examples in~\cite{Jubb:2022,EncycMP_Measurement_in_QFT_FewsterVerch2025}).
	\item there are no superluminal communications associated with the state update rules in the measurement framework.
\end{itemize}
It seems reasonable that the lesson of Sorkin's example is that not all local algebra elements can be implemented as local operations. We return later to whether they can be regarded as local observables.

We are left with the task of proving the lemma.
\begin{proof}[Proof of Lemma~\ref{lem:Theta_on_Lplus}]
Consider the following diagram
\[
\begin{tikzcd}
	\Uf(M)\arrow[ddr,"\tau^+",swap] & \Uf(M;L^+)\arrow[l,"\upsilon_{M;L^+}",swap]\arrow[d,"\cong"] & \Uf(M;L^-)\arrow[d,"\cong"]\arrow[r,"\upsilon_{M;L^-}"] & \Uf(M)\\
	& \Cf(M;L^+) \arrow[r]\arrow[d] & \Cf(M;D(L^-))=\Cf(M;L^-)\arrow[dl] & \\
	& \Cf(M)\arrow[bend right=30,uurr,"(\tau^-)^{-1}",swap]
\end{tikzcd}
\]
in which the left half triangle is a simplified version of one
that appeared in the proof of Lemma~\ref{lem:Theta_on_Kperp}, but with 
$L^+$ replacing $K^\perp$ (valid because $L^+\subset M^+$ \emph{exercise -- check!}). The middle triangle commutes due to transitivity of inclusion, and the equality
$\Cf(M;D(L^-)) =\Cf(M;L^-)$ follows from the timeslice property. The remaining portion is equivalent to the left half triangle but for $L^-\subset M^-$, and also commutes. The inclusions of $\Uf(M;L^\pm)$ in $\Uf(M)$ have been labelled. As the diagram commutes in full,
\[
\Theta\circ \upsilon_{M;L^+} = \upsilon_{M;L^-} \circ \zeta,
\]
where $\zeta$ is obtained by following the uppermost path from 
$\Uf(M;L^+)$ to $\Uf(M;L^-)$ in the diagram. But what this means
is that the range of $\Theta$ restricted to $\Uf(M;L^+)$ is in the image of the inclusion of $\Uf(M;L^-)$ in $\Uf(M)$. In other words,~\eqref{eq:Theta_Lplus} is proved.
\end{proof}

\clearpage
\section{Examples}\label{sect:examples}

We return to the models discussed in the first lecture, but taking a broader viewpoint. 

\paragraph{RFHGHOs} (See., e.g.~\cite{Fewster:2025a}.)
Let $M$ be a globally hyperbolic spacetime.
Let $B$ be a finite rank complex bundle over $M$, equipped with a (possibly indefinite) hermitian fibre metric $(\cdot,\cdot)$. Those less familiar with bundles can think of the example $B=M\times\CC^k$ with the usual complex inner product on $\CC^k$. By $C^\infty(B)$ we mean the smooth sections of $B$ [or smooth functions $C^\infty(M;\CC^k)$ in the simple example].
Suppose also that $B$ has a complex conjugation $\Cc$, i.e., an antilinear involution compatible with the fibre metric in the sense that
\[
(\Cc f,\Cc h) = \overline{(f,h)}.
\]
A partial differential operator $P$ on $C^\infty(B)$ is {\bf real} if $P\Cc = \Cc P$, and {\bf formally hermitian} if 
\[
\int_M (f,Ph)\,\dvol = \int_M (Pf,h)\,\dvol
\]
for all $f,h\in C^\infty(B)$ with compactly intersecting supports.

If a real, formally hermitian operator $P$ also has advanced and retarded Green operators $E^\pm_P$ obeying
\begin{enumerate}[G1]
	\item $E_P^\pm Pf=f$
	\item $PE_P^\pm f=f$ 
	\item $\supp E_P^\pm f\subset J^\pm(\supp f)$ 
\end{enumerate}
for all $f\in C_0^\infty(B)$ then $P$ is a {\bf real formally hermitian Green hyperbolic operator} (RFHGHO).
The free Klein--Gordon field is a RFHGHO (\emph{exercise!}), as are multicomponent generalisations thereof, and other models, including the Proca field.

\paragraph{Quantisation} If $P$ is a RFHGHO on $B$, then let $\Af_P(M)$ be
the unital $*$-algebra generated by symbols $\Phi_P(f)$ ($f\in C_0^\infty(B)$), subject to relations
\begin{enumerate}[Q1]
	\item $f\mapsto \Phi_P(f)$ is complex-linear
	\item $\Phi_P(f)^*=\Phi_P(\Cc f)$
	\item $\Phi_P(Pf)=0$
	\item $[\Phi_P(f),\Phi_P(h)]=\ii E_P(f,h)\1$
\end{enumerate}
for all $f,h\in C_0^\infty(B)$, where
\[
E(f,h) = \int_M (\Cc f, Eh)\dvol.
\]
Local algebras $\Af_P(M;O)$ are formed in the same way as before; the Haag property is proved by an argument in~\cite{FewVer_QFLM:2018}.

If $P$ and $Q$ are RFHGHOs on bundles $B_P$ and $B_Q$ then 
$P\oplus Q$ is a RFHGHO on $B_P\oplus B_Q$ (fibrewise direct sum, with the direct sum complex conjugation and fibre metric) and
\[
E_{P\oplus Q}^\pm  = \begin{pmatrix} E_P^\pm & 0 \\ 0 & E_Q^\pm\end{pmatrix}.
\]
{\bf Fact:} $\Af_{P\oplus Q}(M;O) = \Af_P(M;O)\otimes \Af_Q(M;O)$ under the identification
\[
\Phi_{P\oplus Q}(f\oplus h) = \Phi_P(f)\otimes\1 + \1\otimes \Phi_Q(h).
\]
In other words, the uncoupled combination of $\Af_P(M)$ and $\Af_Q(M)$ is
just $\Af_{P\oplus Q}(M)$.\footnote{The same is true for the Weyl algebra formulation of the scalar field
in terms of $C^*$-algebras. In this case $\otimes$ denotes any complete tensor product, because the Weyl algebra is nuclear.}
\begin{thm}
$K$ be any temporally compact set, and
$T$ be any RFHGHO on $B_P\oplus B_Q$ such that $T=P\oplus Q$ outside $K$. Then $\Af_T(M)$ is a coupled combination of $\Af_P(M)$ and $\Af_Q(M)$.
Moreover, the scattering map is given explicitly by 
\[
\Theta \Phi_{P\oplus Q}(f\oplus h) = \Phi_{P\oplus Q}(\vartheta (f\oplus h)),
\]
where
\begin{equation}\label{eq:Theta} 
\vartheta\begin{pmatrix} f \\ h\end{pmatrix}:= 
\begin{pmatrix}
	f\\ h
\end{pmatrix} 
- (T-P\oplus Q)E_T^{-}
\begin{pmatrix}
	f\\ h
\end{pmatrix} .
\end{equation}
\end{thm}
\begin{proof}
	See~\cite{FewVer_QFLM:2018}.
\end{proof}
Note that the second term in~\eqref{eq:Theta} is supported in the intersection of $K$ with the causal past of the support of $f\oplus h$, and is hence compactly supported (within $K$).  

Now suppose $\Af_P$ is the system and and $\Af_Q$ is the probe theory. 
Consider a probe observable $\Phi_Q(h)$ for some $h\in C_0^\infty(B_Q)$, corresponding to $\Phi_{P\oplus Q}(0\oplus h)$ in the uncoupled combination. Write
\[
\begin{pmatrix}
	f^- \\ h^- 
\end{pmatrix} = \vartheta \begin{pmatrix}
0 \\ h
\end{pmatrix}.
\]
Then we can organise the action of $\Theta$ on powers of $\Phi_{P\oplus Q}(0\oplus h)$ using a generating function
\begin{align*}
\Theta \exp (\ii \Phi_{P\oplus Q}(0\oplus h)) &= 
 \exp (\ii \Theta \Phi_{P\oplus Q}(0\oplus h)) = 
  \exp (\ii \Phi_{P\oplus Q}(f^-\oplus h^-)) \\ &= 
  \exp(\ii \Phi_P(f^-))\otimes  \exp(\ii \Phi_Q(h^-)).
\end{align*}
These expressions are understood as equalities of formal power series in $h$, where both $f^-$ and $h^-$ are $O(h)$. 

Now let $\sigma$ be a probe preparation state for $\Af_Q(M)$. Then
\[
\varepsilon_\sigma( \exp(\ii \Phi_Q(h)) = \eta_\sigma \left(
\exp(\ii \Phi_P(f^-))\otimes  \exp(\ii \Phi_Q(h^-)) \right)= 
\sigma(\exp(\ii \Phi_Q(h^-))) \exp(\ii \Phi_P(f^-)),
\]
which allows us to find $\varepsilon_\sigma(\Phi_Q(h)^k)$ for any $k$ by expanding to find the term of $O(h^k)$,
\begin{align*}
	\varepsilon_\sigma(\1) &= \1 \\
	\varepsilon_\sigma(\Phi_Q(h)) &= \Phi_P(f^-) + \sigma(\Phi_Q(h^-))\1 \\
	\varepsilon_\sigma(\Phi_Q(h)^2) &= \Phi_P(f^-)^2 +  
	2\sigma(\Phi_Q(h^-))\Phi_P(f^-) + 
	\sigma(\Phi_Q(h^-)^2)\1,
\end{align*}
and so on. The impact of the preparation state is seen clearly,
as is the fact that $\varepsilon_\sigma$ is not a homomorphism.

Another consequence is that
\[
\omega(\varepsilon_\sigma( \exp(\ii \Phi_Q(h)) )
 = 
\omega(\exp(\ii \Phi_P(f^-))) \sigma(\exp(\ii \Phi_Q(h^-))) .
\]
Suppose $\omega$ is regular enough that  
\[
\lambda\mapsto \omega(\exp(\ii \lambda \Phi_P(f))) = 
\int \ee^{\ii\lambda \varphi}\dd\mu_f(\varphi) = \widehat{\mu}_{f}(\lambda)
\]
for a unique probability measure $\mu_f$, in which case $\mu_f$ is the
distribution of outcomes of individual measurements of $\Phi_P(f)$ in state $\omega$. Suppose $\sigma$ has a similar property, with measure $\nu_h$. Then
\[
\lambda \mapsto \omega(\varepsilon_\sigma( \exp(\ii \Phi_Q(h)) )= 
\omega(\exp(\ii \lambda \Phi_P(f^-)))\sigma(\exp(\ii\lambda \Phi_Q(h^-)))
= \widehat{\mu}_{f^-}(\lambda)\widehat{\nu}_{h^-}(\lambda)
= \widehat{\rho}(\lambda),
\]
where 
\[\rho=\mu_{f^-}\star\nu_{h^-}
\] 
is the distribution of the sum of the two independent random variables describing the distributions 
of $\Phi_P(f^-)$ and $\Phi_Q(h^-)$ in states $\omega$ and $\sigma$. 
In other words, the measurement scheme introduces {\bf detector noise}. 

\paragraph{Weak coupling} Now suppose that $T=P\oplus Q+\lambda R$ for some coupling constant $\lambda$. 
We have
\[
(P\oplus Q+\lambda R)E_{T}^-f = f,
\]
and hence
\[
(P\oplus Q)E_{T}^-f = f- \lambda R E_T^- f.
\]
Considering supports, and noting that $RE_T^-f$ is compactly supported,
\[
E_{T}^-f = E_{P\oplus Q}^- f - \lambda E_{P\oplus Q}^-R E_T^- f = 
E_{P\oplus Q}^- f - \lambda E_{P\oplus Q}^-R E_{P\oplus Q}^- f +O(\lambda^2).
\]
This is the start of the {\bf Born expansion}. We see that
\[
\begin{pmatrix}
	f^- \\ h^- 
\end{pmatrix} =  \begin{pmatrix}
	0 \\ h
\end{pmatrix} - \lambda R E_{P\oplus Q}^- \begin{pmatrix}
0 \\ h
\end{pmatrix} + O(\lambda^2) =  \begin{pmatrix}
0 \\ h
\end{pmatrix} - \lambda R  \begin{pmatrix}
0 \\ E_Q^-h
\end{pmatrix} + O(\lambda^2).
\]
In particular, if we replace $h$ by $h/\lambda$ and take $\lambda \to 0$, we find
\[
f^- \to f_0:=-R_{PQ} E_Q^-h ,
\]
where $R_{PQ}$ is the $12$ component in the block decomposition of $R$ on $C^\infty(B_P)\oplus C^\infty(B_Q)$. 

By carefully designing $R$ and $h$, we may obtain a desired $f_0\in C_0^\infty(B_P)$ in the limit. 
This is an {\bf asymptotic measurement scheme} for $\Phi_{P}(f_0)$~\cite{FewsterJubbRuep:2022}.
The method can be extended to any algebra element in $\Af_P(M)$, demonstrating that the
measurement framework is comprehensive, as well as causal and consistent. 
Indeed, all elements of $\Af(M;O)$ can be measured (at least asymptotically)
by a measurement scheme with coupling zone $K\subset O$.

\clearpage
\section{Some remarks on interpretation}

At the very least, what has been described in these notes is 
\begin{itemize}
	\item a local protocol for measuring local observables of one QFT (the system)
	using an ancilla (probe) theory, with an effective description at system level in terms of induced observables and state update rules which behave well with respect to composition and causality.
\end{itemize}
As such, it provides
\begin{itemize}
	\item a model for the measurement process in QFT. 
	Here one has to take on board the modelling assumptions, such as the modification of the theory in a (temporally) compact coupling region as a proxy for experimental design. In a real experiment one cannot change the laws of physics, but one certainly is trying to engineer that certain interactions take place preferentially in a specific spacetime region. Another aspect of the modelling is the role of single-use probes, as an idealisation for the way devices are reset (or reset themselves) between measurements. Finally, there remains the issue that measurement of one field is `explained' in terms of measurement of another. 
	Naturally, the probe field can be measured via further probes and the formalism can accommodate that without modification -- see, e.g.~\cite{MandryschNavascues:2024}. This appears to risk an infinite regression, but in actual fact experimenters tell us that they make measurements! What exactly this means is part of the measurement problem, but the measurement chain is finite, albeit possibly long.
\end{itemize}
Viewed in either of the guises above, the framework provides
\begin{itemize}
	\item formal rules for measurement updates in QFT, free of causal pathologies such as `impossible measurements', and 
	\item an operational interpretation for AQFT, in which local algebras $\Af(M;O)$ are best understood as algebras whose self-adjoint elements can (at least approximately) be measured by local means within $O$, rather than as algebras whose unitary elements represent state-update operations that can be performed by local means in $O$ (though of course such operations would be expected to be elements of $\Af(M;O)$). This can facilitate operational treatments of other issues, as for example in~\cite{Ruep:2021,FewsterJanssenLoveridgeRejznerWaldron:2024}.
\end{itemize}

In its current form, the measurement framework does not provide a full resolution of the measurement problem(s) in QFT, however defined, though it can be hoped that it will help to clarify discussions in that direction. In particular, QFT in curved spacetimes may be a better setting to discuss the quantum measurement problem than nonrelativistic quantum mechanics, in which issues such as measurement at spacelike separation are obscured and have to be reintroduced by hand. Similarly, the curved spacetime theory forces one to abandon the crutch of global inertial coordinates available in Minkowski spacetime. Our discussion of multiple state updates can also inform the interpretation (e.g., realist versus instrumentalist) placed on states and the state update process.

\paragraph{Some further reading} The main sources for these notes are the original papers~\cite{FewVer_QFLM:2018},~\cite{BostelmannFewsterRuep:2020} and~\cite{FewsterJubbRuep:2022} and the details omitted here can be found there, including proofs/references for the properties of globally hyperbolic spacetimes that we have skipped. More recent developments briefly mentioned are in~\cite{Fewster:2025a} and~\cite{MandryschNavascues:2024}. 
For more general reading on AQFT in flat and curved spacetimes, see~\cite{AdvAQFT} and Haag's classic book~\cite{Haag:book} -- a pedagogical introduction is given in~\cite{FewsterRejzner_AQFT:2019}. The existence and properties of Green operators for normally hyperbolic operators in curved spacetimes, including the Klein--Gordon operator, is given in~\cite{BarGinouxPfaffle} (some important conventions differ from those here). For a slightly broader survey of measurement in QFT, see~\cite{EncycMP_Measurement_in_QFT_FewsterVerch2025}, and for historical notes on the development of measurement theory see~\cite{FraserPapageorgiou:2023}. 
For quantum measurement theory in the nonrelativistic context, see~\cite{Busch_etal:quantum_measurement}, which particularly influenced our treatment of state updates.
 
\end{document}